\documentclass{ams1}
\usepackage{amsmath,amssymb}
  \usepackage{paralist}
  \usepackage{graphics} 
  \usepackage{epsfig} 
 \usepackage[colorlinks=true]{hyperref}
   
\hypersetup{urlcolor=blue, citecolor=red}

\newtheorem{theorem}{Theorem}[section]
\newtheorem{corollary}{Corollary}

\newtheorem{proposition}{Proposition}

\theoremstyle{definition}
\newtheorem{definition}[theorem]{Definition}

\newcommand{\Ac}{\ensuremath{\mathcal  A}}
\newcommand{\Cc}{\ensuremath{\mathcal  C}}

\newcommand{\Lc}{\ensuremath{\mathcal  L}}
\newcommand{\Pc}{\ensuremath{\mathcal  P}}

\newcommand{\cb}{\ensuremath{\mathbf{c}}}
\newcommand{\sbf}{\ensuremath{\mathbf{s}}}

\newcommand{\bF}{\ensuremath{\mathbb{F}}}
\newcommand{\bN}{\ensuremath{\mathbb{N}}}

\newcommand{\bS}{\ensuremath{\mathbb{S}}}

\newcommand{\gotha}{\ensuremath{\mathfrak{a}}}
\newcommand{\gothb}{\ensuremath{\mathfrak{b}}}
\newcommand{\gothc}{\ensuremath{\mathfrak{c}}}

\newcommand{\gothp}{\ensuremath{\mathfrak{p}}}
\newcommand{\gothq}{\ensuremath{\mathfrak{q}}}

\newcounter{rno}
\newenvironment{rlist}{
\begin{list}{{\normalfont(\roman{rno})}}
{
\setlength{\topsep}{0.25ex}
\usecounter{rno}
\setlength{\topsep}{0.75ex}
\setlength{\labelwidth}{3ex}
\setlength{\leftmargin}{4ex}
\setlength{\labelsep}{1ex}	
\setlength{\rightmargin}{0ex}
\setlength{\itemindent}{0ex}
\setlength{\parsep}{0ex}
\setlength{\itemsep}{0.5ex plus0.2ex minus0.1ex}
}
}{\end{list}}

\newcounter{arno}
\newenvironment{arlist}{
\begin{list}{{\normalfont(\arabic{arno})\hfill}}
{
\usecounter{arno}
\setlength{\topsep}{0.75ex}
\setlength{\labelwidth}{3ex}
\setlength{\leftmargin}{4ex}
\setlength{\labelsep}{1ex}	
\setlength{\rightmargin}{0ex}
\setlength{\itemindent}{0ex}
\setlength{\parsep}{0ex}
\setlength{\itemsep}{0.5ex plus0.2ex minus0.1ex}
}
}{\end{list}}

\title[Codes on Lattices for Random SAF Routing ]
      {Codes on Lattices for Random SAF Routing}


 \email{aghatak@ece.iisc.ernet.in}

\begin{document}
\maketitle

\centerline{\scshape Anirban Ghatak }
\medskip
{\footnotesize
 \centerline{Department of Electrical Communication Engineering}
   \centerline{Indian Institute of Science, Bangalore, India}
}

\medskip


\bigskip

\begin{abstract}
We present a construction of constant weight codes based on the unique decomposition of elements in lattices. The conditions for unique primary decomposition and unique irreducible decomposition in lattices are discussed and the connections with decomposition of ideals in Noetherian commutative rings established. The source alphabet in our proposed construction is a set of uniquely decomposable elements constructed from a chosen subset of irreducible or primary elements of the appropriate lattice. The distance function between two lattice elements is  based on the symmetric distance between sets of constituent elements. It is known that constructing such constant weight codes is equivalent to constructing a Johnson graph with appropriate parameters. Some bounds on the code sizes are also presented and a method to obtain codes of optimal size, utilizing the Johnson graph description of the codes, is discussed. As an application we show how these codes can be used for error and erasure correction in random networks employing store-and-forward (SAF) routing. 
\end{abstract}

\section{Introduction}
Constant weight codes and their bounds have been extensively studied in literature (\cite{Florence}, \cite{Vera}, \cite{EV} and the references therein). Binary constant weight codes have recently been used for error correction in random networks \cite{Maxgad}. The latter involves a choice of the subsets of a given finite set as the code. If the chosen subsets are all of the same cardinality, the resulting code is a constant weight code. This code may also be described as an association scheme, the Johnson scheme (\cite{Florence}, \cite{BCN}). The metric for this code may either be the general Hamming metric or the Johnson metric - as the ``distance" between any two codewords is an even number.\\
Properties of constant weight binary codes on power set lattices were discussed in \cite{MB}. In this paper we exhibit a similar construction of constant weight codes on the elements of \emph{multiplicative lattices} \cite{WD} - i.e. lattices on which a residuation operation and its associated multiplication are defined. If all the non-unit elements can be expressed as finite irredundant  of certain elements of the lattice (termed `primary' elements), the multiplicative lattice is called a \emph{Noether lattice} \cite{WD}, \cite{Dil}, \cite{Dilideal}. Noether lattices are lattice analogues of Noetherian rings - commutative rings with an ascending chain condition. Constant weight codes were constructed in \cite{AG} on the ideals of a Noetherian ring, which were suitable for network error correction in delay-free, acyclic random networks employing store and forward (SAF) routing. We now adapt this construction for Noether lattices and show that the ideal-based codes are a special case of the lattice codes.\\
Kendziorra and Schmidt \cite{KS} interpreted the subspace codes of K\"{o}tter and Kschischang \cite{KK} for error and erasure correction in random networks as codes constructed on the subspace lattice of a finite dimensional vector space over $ \bF_q $ (also treated in \cite{MB}). They further generalized this construction to the submodule lattices of arbitrary finite modules and computed some bounds on code sizes. Their method exploits the additive structure of modules and the unique linear decomposition of each element of a modular lattice in terms of the `independent' spanning set of the lattice. We focus on the multiplicative structure on the ideals of a ring and on the elements of lattices with unique decomposition with respect to meet. It can be shown in this context that the lattice of ideals of an arbitrary commutative ring is a modular lattice. Moreover, if the ring is Noetherian, the ideal lattice is a modular Noether lattice, with decomposition of elements in terms of primary elements. So the properties of codes constructed on multiplicative lattices in general and Noether lattices in particular, will be inherited by the codes based on the ideal structure of commutative rings.\\
Unlike the case of commutative rings, there exist separate conditions for the existence of unique irredundant decomposition of each lattice element in terms of irreducible elements \cite{Dil}, which again does not presuppose modularity. Hence for these lattices having unique irreducible decomposition, the same strategy can be used for construction of constant weight codes, albeit in terms of irreducible elements rather than the primary elements.\\
In the SAF routing strategy for data transmission in networks the intermediate nodes retransmit the received packets without any combination. For correcting packet loss, the existing techniques include ARQ and using erasure codes based on Reed-Solomon codes \cite{Non}, \cite{Xu}. Binary error-correcting codes were proposed in \cite{Maxgad} where a modified Hamming metric was defined to account for adversarial modifications in the network. A recent work \cite{Matroid} proposes codes for SAF where the source transmits constant dimensional flats of a matroid.\\
We propose a $q$-ary error correction scheme for SAF where the source alphabet is a collection of elements of a lattice with unique decomposition. We consider a single source unicast over a delay-free, acyclic network where no knowledge of the network topology is assumed other than having an upper bound on the in-degree at each node in the network. Hence the model is essentially the same as the random network SAF considered in \cite{Matroid}, \cite{Maxgad}. At each non-source node of the network a simple processing, equivalent to a binary search algorithm, is performed before retransmission. The upper bound on the number of incoming nodes is used to provide an estimate of the complexity of this search algorithm. Error correction is achieved in an adversarial model, i.e. we consider the case when corrupt packets are introduced at some network node or nodes and at the destination, the decoder uses a suitable metric to extract the information sent by the source.\\
In section \ref{sec:prelim} we give an overview of multiplicative lattices vis-$\grave{a}$-vis commutative rings. The next section deals with irredundant decompositions of elements in multiplicative lattices and lattices with unique irreducible decompositions. Section \ref{sec:constwt} contains a construction of constant weight codes with a suitable distance metric. Some upper bounds and a lower bound on the size of the codes are given in Section \ref{sec:bnd}. In Section \ref{sec:saf} constant weight lattice codes are used for error and erasure correction for SAF routing in random networks. We conclude with a discussion on the scope of the present scheme and possible improvements.

\section{Preliminary Concepts}\label{sec:prelim}
In this section we give an overview of multiplicative lattices and introduce some definitions and concepts. The conditions for the existence of irredundant primary decomposition in multiplicative lattices are similar to those  for the ideals in Noetherian commutative rings; hence we state the analogous concepts in the case of commutative rings for comparison.\\
\emph{Notation}: We use the standard notations $\bN,\,{\bF}_q$ for the set of natural numbers and the finite field with $q$ elements, respectively. For a set $A$, $\lvert A \rvert$ denotes the cardinality. The gothic letters $\gotha,\, \gothp$ etc. are used to denote the ideals of a ring. The ideal generated by an element $f$ of a ring $R$ is denoted by $(f)$; for instance, the unit ideal by $ (1) $. For two ideals $ \gotha,\, \gothb $ of a ring, $ \gotha +\gothb $ denotes the smallest ideal containing both $ \gotha $ and $ \gothb $. 
\subsection{Multiplicative Lattices}
To characterize multiplicative lattices, we first define the operation of residuation in lattices and list some relevant properties \cite{WD}.
\begin{definition}
Let $\Lc$ be a lattice with unit element $\mathbf{1}$. Then $\Lc$ is said to be \emph{residuated} if there exists a well-defined binary operation $x:y$ for all $x,y \in \Lc$ satisfying the following properties.
\begin{arlist}
\item $a,b \in \Lc \Rightarrow a:b \in \Lc$;
\item $ a:b = \mathbf{1} $ if and only if $ a\geq b $;
\item $ a\geq b \Rightarrow a:c \geq b:c $ and $c:a \leq c:b,\,\, c \in \Lc$;
\item $(a:b):c = (a:c):b, \,\, a,b,c \in \Lc$;
\item $(a \wedge b): c = (a:c)\wedge (b:c)$ and $c: (a\vee b) = (c:a)\wedge (c:b), \,\, a,b,c \in \Lc$.
\end{arlist}.
\end{definition}
\vspace{-3 mm}
In a residuated lattice, there exists an associated multiplication with the following condition: $\Lc$ is completely closed with respect to union and the products of the unions of any two sets of elements in $\Lc$ equals the union of the products of all pairs of elements in the sets. Hence in conjunction with multiplication, the operation of residuation in lattices is governed by the following defining relations:
\begin{rlist}
\item $a \geq (a:b)b\,$;    
\item If $a \geq xb$ then $a:b\geq x$.
\end{rlist}
There is an evident analogy between residuation in multiplicative lattices and the operation of ideal quotients in commutative rings: for ideals $\gotha$ and $\gothb$ in a commutative ring $R$, the ideal quotient $\gotha:\gothb := \{ x \in R \,\lvert \, {\gothb}x \subseteq \gotha \}$. A comparison of some definitions and concepts for multiplicative lattices with their analogues in commutative rings is presented in Table $1$. We conclude this subsection with two concepts related to primary decomposition of elements in lattices.
\begin{definition}
A meet of primary elements $ \bigwedge_i q_i $ is said to be \emph{simple} if $ q_i \ngeq \bigwedge_{j\neq i} q_j\,\, \forall \, i $. 
\end{definition}
\begin{definition}
A prime element $p$ is said to be an \emph{associated prime} of a primary element $q$ if $p\geq q$ and $p\geq b \Rightarrow q \geq b^s, \,\, s \in \bN$.
\end{definition}
\begin{table}\label{Comp}
\begin{tabular}{|l|l|}

\hline
$\mathbf{Multiplicative \,\, Lattice}$ & $\mathbf{Commutative\,\, Ring}$\\
\hline
Residuation & Ideal quotient\\
\hline
\emph{Irreducible element}: & \emph{Irreducible ideal}: \\
$c = d\wedge e \Rightarrow c = d $ or $c=e$ & $\gotha = \gothp \cap \gothq \Rightarrow \gotha =\gothp$ or $\gotha =\gothq$\\
\hline
\emph{Prime element}: & \emph{Prime ideal}:\\
$ p\geq ab \Rightarrow p\geq a $ or $ p\geq b $ & $ ab \in \gothp \Rightarrow a \in \gothp$ or $ b \in \gothp $\\
\hline
\emph{Primary element}: & \emph{Primary ideal}:\\
$ p\geq ab, p \ngeq a \Rightarrow p \geq b^s, \, s\in \bN $ & $ ab \in \gothq, a \notin \gothq \Rightarrow b^s \in \gothq,\, s \in \bN$\\
\hline
\end{tabular}
\vspace{1.5 mm}
\caption{Comparison of Lattices and Rings}
\end{table}
\subsection{Primary Decomposition in Noetherian Rings}
We present an overview of primary decomposition of ideals in Noetherian commutative rings and state some relevant results (\cite{Zariski}, \cite{Atiyah}). 
\begin{definition}

An \emph{ideal} $\gotha$ of a ring $R$ is a subset of $R$ which is an additive subgroup and satisfies: $R \gotha \subseteq \gotha $. The \emph{radical} of an ideal $\gotha$, denoted $\sqrt{\gotha}$, is the set: $\left \{ x \in R \,\lvert \, x^{n} \in \gotha, n \in \bN  \right \}$. If $\gotha = \sqrt{\gotha}$, then $\gotha$ is termed a radical ideal.

\end{definition}
The concepts of \emph{irreducible, prime} and \emph{primary} ideals in a commutative ring are summarised in Table $1$. Evidently, every prime ideal is primary, but the converse is not true. If some prime ideal $\gothp$ satisfies: $\gothp= \sqrt{\gothq}$ for a primary ideal $\gothq$, then $\gothq$ is said to be $\gothp$-\emph{primary}.\\
We now state the central theorem about the existence of primary decomposition of ideals in a Noetherian ring.
\begin{theorem}
\emph{(Noether-Lasker Theorem)} In a Noetherian ring, every ideal admits an irredundant representation as a finite intersection of primary ideals. 
\end{theorem}
From the Noether-Lasker theorem it follows that in a Noetherian ring every ideal $\gotha \subseteq R$ can be expressed as $\gotha = \bigcap^{n}_{i=1} \gothq_i $ for some $n \in \bN$, where $\gothq_i$ are primary ideals such that $\gothq_i \nsubseteq {\bigcap}_{j \neq i} \gothq_j$ for all $1 \leq i \leq n$. For a radical ideal in an arbitrary ring, one can derive sharper results as given by the following theorem.
\begin{theorem}\label{th:rad}\cite{Zariski}
Let $R$ be a ring and $\gotha \subseteq R$ an ideal which admits an irredundant primary decomposition: $\gotha= \bigcap_{i} \gothq_{i}$. Then $\gotha$ is a radical ideal if and only if all $\gothq_{i}$ are prime ideals.
\end{theorem}
From the preceding two theorems one can infer that in a Noetherian ring any radical ideal admits an irredundant decomposition where all the components are prime ideals. The uniqueness of this decomposition is established as follows.
\begin{definition}
Given an irredundant primary decomposition of an ideal $\gotha$ as $\gotha= \bigcap_{i} {\gothq}_{i}$, the prime ideals ${\gothp}_{i}= \sqrt{{\gothq}_{i}}$ are called the \emph{associated} prime ideals of the ideal  $\gotha$. The minimal elements among the set of associated prime ideals of $\gotha$ are called the \emph{isolated} prime ideals of $\gotha$; the others are called \emph{embedded} prime ideals.The primary components $\gothq_{i}$ corresponding to the isolated (embedded) prime ideals of $\gotha$ are called \emph{isolated} (\emph{embedded}) \emph{primary components} of $\gotha$.
\end{definition}
It can be proved \cite{Zariski} that in an arbitrary ring, for an ideal $\gotha$ with an irredundant primary decomposition, all the isolated primary components are uniquely determined by $\gotha$. Moreover, the radical of $\gotha$ is precisely the intersection of the isolated prime ideals of $\gotha$. As a consequence of the Noether-Lasker Theorem and Theorem \ref{th:rad}, it follows that we can uniquely identify every radical ideal of a Noetherian ring as a finite intersection of prime ideals.
\section{Unique Decomposition in Lattices}\label{sec:latdec}
We discuss the conditions for existence of unique decompositions in lattices in terms of prime elements and irreducible elements. Prime element decomposition is a consequence of irredundant primary decomposition in a large class of lattices with ascending chain condition and the relevant results are restatements of those pertaining to primary decomposition in Noetherian rings \cite{WD}. Next the necessary and sufficient conditions for the existence of unique irredundant decomposition in terms of irreducible elements in a lattice are stated and compared with the conditions for primary decomposition. 
\subsection{Primary Decomposition in Lattices}
In \cite{WD}, \emph{Noether lattices} were defined in connection with primary decomposition in multiplicative lattices as follows.
\begin{definition}
A lattice $\Lc$ is called a \emph{Noether lattice} if it satisfies the following conditions:
\begin{arlist}
\item $\Lc$ may be residuated;
\item $ \Lc $ has ascending chain condition;
\item Every irreducible element of $\Lc$ is primary.
\end{arlist} 
\end{definition}
It was observed in \cite{WD} that all the decomposition theorems proved for the ideals of a Noetherian commutative ring hold for Noether lattices. Specifically, every non-unit element of a Noether lattice can be expressed as a simple finite meet of primary elements, each associated with a different prime element. Moreover, the existence of such an irredundant primary decomposition was stated to be independent of modularity - an example of a non-modular Noether lattice was given in \cite{WD} (p.~$ 351 $). However, in a subsequent work \cite{Dilideal}, the definition of Noether lattices was modified to include modularity, alongwith a condition to ensure that all irreducible elements were primary. It was also remarked that the lattice of ideals of a Noetherian ring is a (modular) Noether lattice. But while the assumption of modularity was necessary to formulate the lattice analogues of deeper results in commutative ring theory, for the existence of irredundant primary decomposition in multiplicative lattices the first definition of a Noether lattice is sufficient. So for the remainder of the paper we stick to the original definition of Noether lattices and refer to the latter as \emph{modular} Noether lattices.\\
We now establish the lattice analogue of the unique decomposition of radical ideals in terms of prime ideals in a Noetherian ring. First we define the analogue of a radical ideal.
\begin{definition}
Let $ a \in \Lc $ where $ \Lc $ is a multiplicative lattice. The \emph{radical} of $ a $, denoted $ \sqrt{a} $, is defined as $ \sqrt{a}:= \bigwedge q $ for all $ q \in \Lc $ which satisfy: $ a \leq q $ and $ b \leq q $ for all $ b \in \Lc $ such that $ b^s \leq a $ for some $ s \in \bN $. If $ a = \sqrt{a} $, it is termed a \emph{radical element}.
\end{definition}
The following proposition is important in the context of unique decomposition in terms of prime elements.
\begin{proposition}\label{prop:radint}
The radical of a finite meet in a multiplicative lattice with ascending chain condition is the meet of the radicals of the elements.
\begin{proof}
If $ b \leq \sqrt{\bigwedge_i a_i} $, then we have $ b^s \leq \bigwedge_i a_i  $, i.e. $ b^s \leq a_i, \, \forall i \Rightarrow b \leq \sqrt{a_i},\, \forall i $. Therefore, $ \sqrt{\bigwedge_i a_i} \leq \bigwedge_i \sqrt{a_i} $. Conversely, if $ b \leq \bigwedge_i \sqrt{a_i},\,\, \exists \, s_i \in \bN $ such that $ b^{s_i} \leq a_i, \, \forall i$ and so, $ \exists \, s \in \bN $ (vide properties of residuation and multiplication in lattices, \cite{WD}, p.~$ 337 $) such that $ b^s \leq \bigwedge_i a_i $. Hence, $ b \leq  \sqrt{\bigwedge_i a_i} $ and the proof is complete.
\end{proof}
\end{proposition}
From the definition of the radical of an element and that of associated prime elements of a primary element (\cite{WD}) it follows that the radical of a primary element is a minimal prime associated with it. Hence we have the following theorem.
\begin{theorem}\label{th:radprime}
In a Noether lattice, every radical element admits a unique simple decomposition as the meet of prime elements.
\end{theorem}
\begin{proof}
Let $ a \in \Lc $, where $ \Lc $ is a Noether lattice, have a simple primary decomposition: $ a = \bigwedge_i q_i $. Taking the radical on both sides and using Proposition \ref{prop:radint}, we have: $ \sqrt{a} = \bigwedge_i \sqrt{q_i} $. It is known (\cite{WD}, p.~$ 347 $) that the associated prime elements and the number of primary elements in a simple primary decomposition of an element in a Noether lattice are uniquely determined. Hence this is a unique simple decomposition of the radical element $ \sqrt{a} $ in terms of prime elements.
\end{proof}
A significant difference between Noether lattices and Noetherian commutative rings is that in the latter all irreducible ideals are primary as a consequence of acc on the ideals, whereas in the former irreducible elements are not primary per se. A set of sufficient conditions for irreducible elements (with respect to meet) to be primary is given in \cite{Dilideal} (Theorem $ 3.1 $) which includes modularity, acc, and the condition that every element be a union of so-called \emph{meet principal elements}. Moreover, the existence of unique irreducible decomposition in lattices cannot be treated as a subcase of primary decomposition and vice versa.
\subsection{Irreducible Decomposition in Lattices}
The question of the existence of unique decomposition in lattices in terms of irreducible (with respect to meet) elements was settled in \cite{Dil} where a subset of so-called \emph{Birkhoff lattices} was shown to meet the requirement. Birkhoff lattices satisfy the following property: If $a > a\wedge b$ (i.e. $a$ covers $a\wedge b$) then $a \vee b > b\,$; moreover, there exists a \emph{rank function} with the usual properties. These lattices were subsequently termed (upper) \emph{semimodular lattices} \cite{Birk, MS}. We now state the key theorem giving the necessary and sufficient conditions for the existence of unique irreducible decomposition.
\begin{theorem}\label{th:irred}\cite{Dil}
Let $ \Lc $ be a lattice with unit element in which every quotient lattice is of finite dimension. Then each element of $ \Lc $ is uniquely expressed as a meet of irreducibles if and only if $ \Lc $ is a Birkhoff lattice in which every modular sublattice is distributive. Then for every element $ b $ covering an element $ a $, there is exactly one irreducible $q$ such that $ q \ngeq b$, but $q\geq b'$ for all other $b'$ covering $a$. These irreducibles are the components of $a$.
\end{theorem}
Therefore, in one sense the conditions for the existence of unique irreducible decomposition are more restrictive than those for irredundant primary decomposition as in the former case the modular sublattices have to be distributive. On the other hand, the existence of unique irreducible decomposition does not require that all irreducible elements be primary, which is a precondition for a Noether lattice - modular or not. This is illustrated by the first example in \cite{WD} (p.~$ 351 $), reproduced in Figure \ref{fig:Non-Noether}, of a lattice which has acc but is not a Noether lattice.
{\begin{figure}
\begin{center}
    \includegraphics[width=1.0 in]{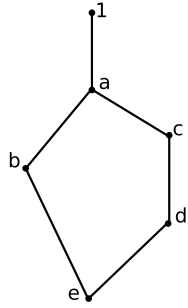}\\
  \caption{Non-Noether lattice with unique irreducible decomposition}\label{fig:Non-Noether}
  \end{center}
\end{figure}}
The lattice is evidently non-modular, as it has the standard five-element non-modular sublattice. The multiplicative relations, as defined in \cite{WD}, on the elements of this lattice may be summarised thus: $ xy = b $ for all $ x, y \in \{a, b \} $; $ xy = e $ for all $ x \in \{ c , d, e\} $ and $ y \in \{a, b, c, d, e \} $ and the unit element is the multiplicative identity. With this multiplication, not all irreducible elements of the lattice are primary: $ d > e = bc $ but $ d \ngeq c $ and $ d \ngeq b^s $ for any $  s \in \bN $ as $ b $ is idempotent. However, it is obvious that every element of the lattice has unique irreducible decomposition.\\
The ideal lattices of Noetherian rings, which are modular lattices, do not possess unique irreducible decomposition in general, as not all are distributive. So a ring-theoretic example is furnished by the ideal lattice of a \emph{Dedekind domain} \cite{Atiyah}, which is a one-dimensional, integrally closed Noetherian integral domain. Every non-zero prime ideal in a Dedekind domain is maximal (since one-dimensional) and every non-zero ideal can be expressed as a unique product of prime ideals. The maximality of every non-zero prime ideal makes it an irreducible element in the ideal lattice ($ \gothp = \gotha \wedge \gothb = \gotha \cap \gothb \Rightarrow \gothp = \gotha\, $ or $ \,\gothp = \gothb $); so the ideal lattice of a Dedekind domain has unique irreducible decomposition. In fact, a Dedekind domain provides a common example for both unique prime decomposition in a modular Noether lattice and unique irreducible decomposition. Hence the distributivity of the ideal lattice of a Dedekind domain, with the operations `$ \gotha \vee \gothb $' and `$ \gotha \wedge \gothb $' defined as $ \gotha +\gothb $  and $ \gotha \cap \gothb $ as outlined in the following proposition, is a direct consequence of Theorem \ref{th:irred}.
\begin{proposition}\label{prop:Dedom} \cite{Atiyah}
Let $\gotha , \, \gothb, \, \gothc  $ be non-zero ideals of a Dedekind domain $R$. Then the following distributive laws are satisfied:
\begin{rlist}
\item $ \gotha \cap (\gothb +\gothc) = (\gotha \cap \gothb) + (\gotha \cap \gothc)$
\item $ \gotha + (\gothb \cap \gothc) = (\gotha + \gothb) \cap (\gotha + \gothc)$
\end{rlist}
\end{proposition}
\section{Constant Weight Codes on Lattices with Unique Decomposition}\label{sec:constwt}
In this section we construct constant weight codes on elements of a lattice with unique decomposition with respect to meet. As discussed in Section \ref{sec:latdec}, there are two methods for unique decomposition of elements in lattices:
\begin{rlist}
\item decomposition of radical elements as unique meet of prime elements in a Noether lattice;
\item irreducible decomposition in a(n) (upper) semimodular lattice with all modular sublattices distributive.
\end{rlist} 
The proposed construction is applicable to either case and so, in the remainder of this paper, we refer to radical elements in a Noether lattice and elements of a lattice with irreducible decomposition as \emph{decomposable elements} and prime elements in the first case alongwith irreducible elements in the second as \emph{constituent elements}. 
\subsection{Encoding and Distance Metric}
We choose a finite subset $\Pc$ of the set of constituent elements of a suitable lattice $\Lc$ with unique decomposition, such that these elements do not divide one another. The source alphabet is denoted by: $\Ac = \left\{ 0, 1, \ldots, M-1 \right\}$, where each element of $\Ac$ represents a distinct meet of a finite number of elements in $\Pc$. The unique decomposition of elements in the lattice ensures that distinct combinations of constituent elements from the set $\Pc$ yield distinct decomposable elements. If each element of $\Ac$ is constructed by the meet of a constant number of elements in $\Pc$, we have a constant weight code based on $\Pc \subseteq \Lc$.\\
For instance, if a decomposable element $q$ has the following decomposition: $q = \bigwedge^{k}_{i=1} p_i$ where $p_i \in \Pc \,\,\forall \,\, i$, we can uniquely identify $q$ with the subset $\left \{ p_1, p_2, \cdots, p_k \right \}$ $ \subset \Pc$. If $\lvert \Pc\rvert = n$ and the chosen subsets are of the same cardinality, say $k$, the encoding procedure is equivalent to constructing a generalized Johnson graph $J(n,\,k,\, i)$ (where $k - i$ is the number of common places of any two adjacent `vertices') on the set $\Pc$, with the vertices corresponding to decomposable elements uniquely determined by each $k$-subset of constituent elements. Adjacent vertices of the graph correspond to codewords which are at minimum distance from each other. In this context we next define the distance between two codewords, which is the distance between two sets of constituent elements. The natural and obvious distance metric is the symmetric distance between arbitrary finite sets.
\begin{definition}
The \emph{symmetric distance} ${\rm d_s}\,(A,B)\,$ between arbitrary sets $A, \,B$ is defined as: ${\rm d_s}\,(A,B) = \lvert A \Delta B \rvert = \lvert (A\backslash\,B) \,\cup \,(B\backslash\,A)\rvert$.
\end{definition}
We now formally define a constant weight lattice code as follows.
\begin{definition}
Let $\Pc $ be a subset of the set of constituent elements of a lattice with unique (meet) decomposition $\Lc$. A code on the lattice is defined as a subset $\Cc$ of $ P_f (\Pc)$, where $ P_f (\Pc)$ is the set of all finite subsets of $\Pc$. If $\lvert \Pc \rvert = n$ and all the chosen elements of $\Cc$ are sets of cardinality $k$, $\Cc$ is a constant weight code, denoted by $\Cc (n,k)$, with each codeword of weight $k$.
\end{definition}
It follows that the minimum distance of a constant weight lattice code is given by: 
\begin{equation*}
\mathrm{d_{min}}(\Cc) = \mathrm{min}_{\begin{subarray} {c} 
c,c' \in \Cc \\ c \neq c'
\end{subarray}}\,\, {\rm d_s} (c,\, c') 
\end{equation*}
A code consisting of $k$-subsets of the set of chosen constituent elements $\Pc$, with $\lvert \Pc \rvert = n$, and minimum distance ${\rm{d_{min}}} = d$, will be termed an $(n, \,k,\, d)$ code.  
\subsection{Decoding}
In the absence of errors, the decoder receives a subset of constituent elements and reconstructs the transmitted lattice element by computing the meet. The decoder makes an estimate of the transmitted codeword by a minimum distance decoding rule. For instance, if $\sbf$ is the received subset, a minimum distance decoder will yield the following estimate: $\hat{\cb} = {\rm arg \, min}_{{\cb}' \in \Cc}\,{\rm d_s}(\sbf,\,{\cb}')$.
\subsection{Optimal Constant Weight Codes}
As discussed before, a set of finite subsets of constituent elements of a suitable lattice is chosen as a constant weight code. Hence a constant weight code so constructed may be viewed as a subgraph of a generalized Johnson graph $J(n,\,k,\, i)$ and the code has minimum distance $ {\rm{d_{min}}} = 2i $; it follows that a \emph{maximal} subgraph in the generalized Johnson graph will be an optimal constant weight code. A \emph{clique} in a graph is a complete subgraph which is not contained in any other complete subgraph. Intuitively, an optimal code on a generalized Johnson graph should contain more members than a clique, as all the members of a clique are exactly at minimum distance from each other. However, it turns out that (Corollary \ref{cor:mclique}) for certain parameters an optimal code on the generalized Johnson graph can be realized as a maximal clique. Some optimal codes have been found for small parameters by computer search implementing a variant of the well-known Bron-Kerbosch algorithm \cite{BK}. For instance, an optimal $(8,\,6,\,4)$ constant weight code is a maximal subgraph in the generalized Johnson graph $J(8,\,6,\, 2)$ comprising distinct $ 6 $-element subsets of a set of $ 8 $ elements such that every pair of distinct subsets differ in at least $ 2 $ places. There are $ 105 $ maximal subsets of size $ 4 $; choosing any of them will give an optimal $(8,\,6,\,4)$ constant weight code. For example, labelling the set of constituent elements as $ p_0,\, p_1, \cdots ,\, p_6, \, p_7 $, a constant weight code is given by $\lbrace \lbrace p_0,\, p_1,\, p_2, \, p_3, \, p_4,\, p_6 \rbrace , \, \lbrace p_0,\, p_1,\, p_2, \, p_5, \, p_6,\, p_7 \rbrace , \,\lbrace p_0,\, p_1,\, p_3, \, p_4, \, p_5,\, p_7 \rbrace ,$ $ \lbrace p_2,\, p_3,\, p_4, \, p_5, \, p_6,\, p_7 \rbrace \rbrace$.
\subsection{Illustrative Code on the Lattice of Ideals of a Dedekind Domain}
It was discussed in Section \ref{sec:latdec} that the ideal lattice of a Noetherian commutative ring is an example of a modular Noether lattice and the primary decomposition theorems of ideals can also be viewed as decomposition of arbitrary lattice elements as meets of primary elements. However for the existence of unique decomposition in terms of irreducible elements we require that all modular sublattices be distributive - a condition fulfilled in the ideal lattice of a Dedekind domain. We construct a constant weight lattice code based on irreducible decomposition of elements in the ideal lattice of a Dedekind domain. This can be dually viewed as  an example of decomposition in terms of primary elements as a Dedekind domain is a one-dimensional, integrally closed Noetherian domain.\\
Polynomial rings in one indeterminate over fields are \emph{principal ideal domains} (PID), and so, natural examples of Dedekind domains. We construct a $(7,\,4,\,4)$ constant weight code based on a subset $\Pc$ of the prime spectrum of $\bF_2 [X]$. Prime ideals in any unique factorization domain (and hence, in any PID) are generated by irreducible elements; in the particular case of a polynomial ring over a field, prime ideals are generated by the irreducible polynomials. Hence prime ideal decomposition in this case is an example of irreducible decomposition on the ideal lattice.\\
Consider the following set of irreducible polynomials in $\bF_2 [X]$: \\
$f_1 = X^2 + X + 1,\, f_2 = X^3 + X +1, \, f_3 = X^3 + X^2 +1, \, f_4 = X^4 +X+1, \, f_5 = X^4 + X^3 +1,\, f_6 = X^5 + X^2  +1,\, f_7 = X^6 + X +1$.\\
Define the set $\Pc := \left \{ \gothp_1,\, \gothp_2,\, \gothp_3,\, \gothp_4,\, \gothp_5,\, \gothp_6,\, \gothp_7 \right \}$, where $\gothp_i = (f_i)$, the ideal generated by the irreducible polynomial $f_i,\,\, i = 1,\,2,\,\cdots,\, 7$. A $(7,\,4,\,4)$ code on $\Pc$ will consist of a collection of sets of $4$ distinct prime ideals such that any two sets differ in at least two elements. For ease of notation, we represent each prime ideal by its index; for instance, the set $\left \{ \gothp_1,\, \gothp_2,\, \gothp_3,\, \gothp_4 \right \}$ will be denoted by $\left \{1,\,2,\,3,\,4 \right \}$. A possible optimal $(7,\,4,\,4)$ code is the following collection of sets of prime ideals: $\left \{1,\,2,\,3,\,6 \right \},\, \left \{1,\,2,\,4,\,5 \right \},\, \left \{1,\,3,\,4,\,7\right \},\, \left \{1,\,5,\,6,\,7 \right \}$, $\lbrace 2,\,3,\, 5,\,7 \rbrace$, $\left \{2,\,4,\, 6,\,7\right \},$ $ \left \{3,\,4,\, 5,\, 6 \right \}$.\\
The source alphabet will consist of those radical ideals of $\bF_2 [X]$ which correspond to the intersection of the sets of prime ideals constituting the codewords. In the present example, where the prime ideals are generated by non-associate irreducible polynomials, i.e. which are not multiples of one another by unit elements, one can identify each radical ideal of the source alphabet with a distinct polynomial in $\bF_2 [X]$. The above assertion is a consequence of irreducible decomposition in the ideal lattice of Dedekind domains as stated in the following proposition.
\begin{proposition}\cite{AG}
Let $R$ be a polynomial ring over a field and let $\gothp_1, \gothp_2, \cdots, \gothp_n$ be prime ideals of $R$ such that $\gothp_i = (f_i), \, f_i \in R,\, i=1,\,2,\cdots,\, n, $ are non-associate irreducible polynomials. Then the radical ideal $\gotha = \bigcap^{n}_{i=1} \gothp_i$ is generated by the polynomial $f = \prod^{n}_{i} f_i$.
\end{proposition}
We finally describe the radical ideals which constitute the source alphabet for the $(7,\,4,\,4)$ code given above in terms of their generating polynomials. We write the generating polynomials in the form of hexadecimal representation of their coefficients expressed as a binary string. For example, the generating polynomial for the radical ideal encoded as $\lbrace 1,\,2,\,3,\,6 \rbrace$, given by $X^{13} + X^{11} + X^9 + X^8 + X^5 + X^3 + 1$, has the binary string $ 0010101100101001 $ and is represented as $2B29$. Hence we have: $\Ac = \lbrace (2B29), \, (2E7B), \, (93BD), \, (23D75),\, (144B1),\, (5F237), \, (17153)\rbrace$.
\section{Bounds on the Code Size}\label{sec:bnd}
The proposed constant weight code based on elements of lattices with unique decomposition can be described as a Johnson scheme on a subset of the constituent elements of the lattice. Bounds on the sizes of such constant weight codes are available in literature (\cite{Florence}-\cite{EV}, \cite{SMJ}, \cite{E1}). We state some of these, namely, the sphere-covering \& packing, Singleton and Johnson bounds, in the context of constant weight lattice codes. Our results are, as expected from the discussion in Section \ref{sec:latdec}, mostly restatements of the said bounds derived for constant weight codes constructed on the prime ideals of a Noetherian ring in \cite{AG}, in the generalized setting of lattices with unique decomposition. As such, they share the same relation with the K\"{o}tter-Kschischang \cite{KK} constant dimensional subspace codes as the codes in \cite{AG}, since the Grassmann association scheme defined over the finite field $\bF_{q}$ is a $q$-analogue of the Johnson scheme \cite{BCN}.
\subsection{Sphere-Packing and Sphere-Covering Bounds}
For these bounds, we first define a sphere in the context of constant weight codes on lattices. As before, a code $\Cc (n,\,k,\, d)$ of weight $k$ and minimum distance $d$ is defined on a set $\Pc $ of cardinality $n$; each codeword is a $k$-subset of constituent elements belonging to $\Pc_k$, the set of $ k $-subsets of $ \Pc $.
\begin{definition}
A sphere ${\bS}(\sbf,\,k,\,r)$ of radius $r$ centered at a set $\sbf\in \Pc_k$ is defined as the set of all $k$-subsets of constituent elements ${\sbf}'\in \Pc_k$ such that ${\rm d_s}(\sbf,\, {\sbf}') \leq 2r$.
\end{definition}
The number of $k$-subsets of constituent elements in ${\bS}(\sbf,\,k,\,r) $, where $ k, (n-k) \geq r$, is: $\lvert{\bS}(\sbf,\,k,\,r) \rvert = \sum^{r}_{i=0} \binom{k}{i} \binom{n-k}{i}$. It follows that $\lvert {\bS}(\sbf,\,k,\,r)\rvert$ is independent of the choice of $\sbf \in \Pc_k$ and $\lvert {\bS}(k,\,r)\rvert = \lvert {\bS}(n - k,\,r)\rvert$. We now state the sphere-packing and sphere-covering bounds for the constant weight lattice codes.
\begin{theorem}
Let $\Cc(n,\,k,\, d)$ be defined as before, with ${\rm d_{min}}(\Cc)\, \geq 2r\,$ and let $t = \lfloor \frac {r - 1}{2} \rfloor$. Then
\begin{equation*}
\lvert\Cc(n,\,k,\, d)\rvert \leq \dfrac{\binom {n}{k}}{\lvert {\bS}(k,\,t)\rvert} = \dfrac{\binom {n}{k}}{\sum^{t}_{i=0} \binom{k}{i} \binom{n-k}{i}} 
\end{equation*}
Conversely, there exists a code $\Cc (n,\,k,\, d)$ with  ${\rm d_{min}}(\Cc )\, \geq d\,$ with
\begin{equation*}
 \lvert \Cc \rvert \,\geq \, \dfrac{\binom {n}{k}}{\lvert {\bS}(k,\,t+1)\rvert} = \dfrac{\binom {n}{k}}{\sum^{t+1}_{i=0} \binom{k}{i} \binom{n-k}{i}} 
\end{equation*}
 \end{theorem}
\subsection {Singleton Bound}
To derive a Singleton-type bound for a lattice-based code we define a notion of puncturing. Henceforth, an $(n,\,k,\, d)$ code of size $N$ is referred to as an $(n,\, k,\, N,\, d)$ code. A punctured code ${\Cc}'$ is obtained from $\Cc$ thus:
\begin{rlist}
\item Replace $\Pc$ by a subset ${\Pc}' \subset \Pc$ of cardinality $n -1$. 
\item Replace any codeword $\sbf = \{ p_1 , p_2 , \ldots , p_k \} \subset \Pc $ by the set ${\sbf}' = \sbf \cap {\Pc}'\,$ if $\,\lvert\, {\sbf}'\rvert = k-1$; otherwise replace $\sbf$ by any $(k-1)$-subset ${\sbf}' \subset \sbf$.
\end{rlist}
The above ``puncturing" operation follows the K\"{o}tter-Kschischang formulation for subspace codes and has been applied in \cite{AG} as well, where $ k $-subsets of prime ideals of a Noetherian ring form the constant weight codewords.
\begin{theorem}
If $\Cc$ is a constant weight lattice code of type $(n,\, k,\,N,\, d), \, d > 2,\,$ on a set $\Pc $ of constituent elements of a lattice with unique decomposition, then the code ${\Cc}'$ on the set ${\Pc}'$ obtained by the above puncturing operation is a code of type $(n-1,\, k-1,\,N,\, d')$, where $d \geq d' \geq d -2$.   
\begin{proof}
Let ${\sbf}'_1, {\sbf}'_2 \in {\Cc}'$ be punctured codewords obtained from $\sbf_1, \sbf_2 \in \Cc$. Then we consider the following cases:
\begin{arlist}
\item Let ${\sbf}'_i = \sbf_i \cap {\Pc}',\, i =1,\,2$. Then we can write $\sbf_1 = \{ p_1, \cdots, p_{k-1}, p \}, \, \sbf_2 = \{ q_1, \cdots, q_{k-1}, p \}, $ where $\, \{ p \}= \Pc \backslash \, {\Pc}'$. Therefore ${\rm d_s}({\sbf}'_1,\, {\sbf}'_2) \geq d -2$.

\item Let $ {\sbf}'_1 = \sbf_1 \cap{\Pc}',\,{\sbf}'_2 = \{ q_1, \cdots , q_{k-1} \}, \, \sbf_1 = \{ p_1, \cdots, p_{k-1}, p_k \}, \,\sbf_2 = \{{\sbf}'_2\} \cup \{ q_k \}$ where we have $\{ p_k \} = \Pc \backslash  {\Pc}'$ and $ q_k \in  {\Pc}'$. Then as ${\rm d_s}(\sbf_1,\, \sbf_2) = d$, it follows that $\lvert \,({\sbf}'_1 \, \Delta \, {\sbf}'_2)\,\rvert \geq  d -2$.

\item Let $\sbf_1,\,\sbf_2$ be as above and let ${\sbf}'_1, \,{\sbf}'_2$ be obtained by dropping one arbitrary element from $\sbf_1,\, \sbf_2$ respectively. If ${\rm d_s}(\sbf_1,\, \sbf_2) = d$, then ${\rm d_s}({\sbf}'_1,\, {\sbf}'_2) \geq d -2$.

\end{arlist}
Further, for any two distinct codewords in the original code $ \Cc$, say, $\sbf_1 = \{ p_1, \cdots, p_{k-1},$ $ p_k \},\,\, \sbf_2 = q_1, \cdots , q_{k-1}, q_{k} \}$, we have $\{ p_1, \cdots, p_{k-1} \} \neq \{ q_1, \cdots ,q_{k-1} \}$ as ${\rm d_s}(\sbf_1,\,\sbf_2) \geq d > 2$. Hence $\lvert \, \Cc\, \rvert = \lvert \, {\Cc}'\, \rvert$.
\end{proof}
\end{theorem} 
Now we state a Singleton-type bound for constant weight codes over a lattice.
\begin{theorem}
The cardinality of an $(n,\, k,\,N,\, d)$ code $\Cc$ based on a set $\Pc $ of constituent elements of a lattice with unique decomposition, with minimum distance $d > 2$, has the following upper bound:
\begin{equation*}
N \leq \displaystyle{\binom {n - (d - 2)/2}{ {\rm max}\{k, n - k \}}}.
\end{equation*}
\begin{proof}
Assume both $k,\, n-k \geq (d-2)/2$.\\
If $\Cc$ is punctured as described above $(d-2)/2$ times in succession, a code ${\Cc}_0$ is obtained which is of type $(n -(d-2)/2,\, k - (d-2)/2,\,N,\, d_0)$, where $d_0 \geq 2$. Each codeword of ${\Cc}_0$ consists of a set of $(k-(d-2)/2)$ of a subset ${\Pc}_0$ of the constituent elements, with cardinality $(n- (d-2)/2)$. Hence the size of ${\Cc}_0$ is upper bounded by the number of distinct choices of subsets possible:
\begin{equation*}
\lvert {\Cc}_0 \rvert \leq \displaystyle{\binom {n - (d - 2)/2}{k - (d - 2)/2}}= \displaystyle{\binom {n - (d - 2)/2}{n - k }}. 
\end{equation*}
Similarly, we could begin with the corresponding type $(n ,\, n -k,\,N,\, d_0)$ code ${\Cc}^{\perp}$ and obtain the corresponding $(n -(d-2)/2,\, n -k - (d-2)/2,\,N,\, d_0)$ code ${{\Cc}_0}^{\perp}$ which will have cardinality bounded above by:
\begin{equation*}
\lvert{{\Cc}_0}^{\perp} \rvert \leq \displaystyle{\binom {n - (d - 2)/2}{k }}. 
\end{equation*}
The theorem follows.
\end{proof}
\end{theorem}
\subsection{Johnson Bounds}
We now formulate two Johnson bounds for constant weight codes on lattices with unique decomposition. These bounds are similar to the classical Johnson bounds on constant weight binary codes \cite{Florence, Vera} and so we only outline the proofs.
\begin{theorem}[\,(Restricted)\, Johnson Bound 1]\label{th:J1}
Given an $(n,\, k ,\,N,\, 2\delta)$ lattice code over a set $\Pc $ of constituent elements in a lattice with unique decomposition, the size of the code is upper bounded by $N \leq \biggl \lfloor \, \dfrac{\delta n}{k^2 - kn + \delta n} \, \biggr \rfloor$, given $k^2 - kn + \delta n > 0$.   
\end{theorem}
\emph{Sketch of proof}: Let $\Pc = \lbrace p_1,\, p_2,\, \cdots,\, p_n \rbrace $. Construct an $N \times n$ array $A = (a_{ij})$ with $(0,\,1)$ entries, where each row corresponds to a codeword and the columns correspond to the elements of $\Pc$. Then $a_{ij} = 1$ if the $j$-th element in $\Pc$ belongs to the $i$-th codeword, else $a_{ij} = 0$. Now define: 
\begin{equation*} 
S = \sum^{N}_{i=1}\sum^{N}_{j=1,\, j \neq i}\sum^{n}_{l=1} a_{il}a_{jl}.
\end{equation*}
As any two codewords have difference at least $2\delta$, it follows that $S \leq (k - \delta)N(N-1)$. Again, if the number of non-zero entries in the $j$-th column is denoted $\nu_j$, we have: $\sum ^{n}_{ j=1} \nu_j = kN$ and $S = \sum ^{n}_{j=1} \nu_j (\nu_j - 1)$. Using the Cauchy-Schwartz inequality,  
\begin{equation*}
n ( \sum ^{n}_{ j=1}{ \nu_j}^2 ) \geq \, ( {\sum ^{n}_{ j=1} \nu_j} ) ^2  = k^2 N^2.
\end{equation*} 
Solving for the maximum value of $N$ gives the stated bound when $k^2 - kn + \delta n > 0$.\\

The following simple corollary to the above theorem points out the interesting fact that the restricted Johnson bound, when achievable, is attained by the codes based on the maximal cliques of the corresponding Johnson graph.
\begin{corollary}\label{cor:mclique}
If an $(n,\, k ,\, 2\delta)$ lattice code achieves the restricted Johnson bound, then it can be realized on a maximal clique in the generalized Johnson graph $J(n,\,k,\, \delta)$.
\end{corollary}
\begin{proof}
For $N = \dfrac{\delta n}{k^2 - kn + \delta n} $, we have: $ (k - \delta)N(N-1) = (1/n)k^2 N^2 - kN $.\\

With equality in the Cauchy-Schwartz inequality in the proof of Theorem \ref{th:J1}, i.e. setting $ \sum ^{n}_{ j=1}{ \nu_j}^2  = (1/n)k^2 N^2 $, we have $ S = (k - \delta)N(N-1) $. The second equality is achieved when any two rows in the array $ A $ differ in exactly $ 2 \delta $ places, i.e. the corresponding codewords are adjacent in the Johnson graph $J(n,\,k,\, \delta)$.
\end{proof}
The following theorem is directly adapted from a refinement on the restricted Johnson bound given in \cite{Florence}, which serves as a further check on the attainability of the above bound. The modification involved is in minimizing the sum $ \sum ^{n}_{ j=1}{ \nu_j}^2 $, subject to the constraints that all $ \nu_j $'s are integers and the sum $\sum ^{n}_{ j=1} \nu_j$ is the number of non-zero entries in the array $ A $. In Theorem \ref{th:J1} this was done by directly applying the Cauchy-Schwartz inequality.
\begin{theorem}\label{th:sharpJ1}
For an $(n,\, k ,\,N,\, 2\delta)$ lattice code if there exist integers $ a, b $ satisfying: $ kN = na + b, \, 0\leq b< n $, then: $ na(a-1)+ 2ab \leq (k - \delta)N(N-1) $.
\end{theorem}
\emph{Remark}: For the $ (7,\,4,\,4) $ lattice code in Section \ref{sec:constwt}, we have $ k^2 - kn + \delta n = 16 - 28 + 14 =2 $, and hence the first (restricted) Johnson bound can be applied to give $ N \leq 7 $. Hence the $ (7,\,4,\,4) $ code constructed from a maximal clique in the corresponding generalized Johnson graph meets the Johnson bound.\\
Finally we present the counterpart of the so-called unrestricted Johnson bound, which does not impose the condition $ k^2 - kn + \delta n > 0 $.
\begin{theorem}[\,(Unrestricted)\, Johnson Bound 2]\label{th:J2}
The size of an $(n,\, k ,\,N,\, 2\delta)$ lattice code over a set $\Pc$ of constituent elements of a lattice with unique decomposition, with cardinality $\lvert \Pc \rvert = n$, is upper bounded by 
\begin{equation*}
N \leq \biggl \lfloor \, \dfrac{n}{k} \,\biggl \lfloor \, \dfrac{ n-1}{k-1} \cdots \, \biggl \lfloor \, \dfrac{ n-(k- \delta)}{\delta} \, \biggr \rfloor \cdots  \, \biggr \rfloor \, \biggr \rfloor. 
\end{equation*}.
\end{theorem}
\emph{Sketch of proof}: Denote the maximum size of an $(n,\, k ,\, 2\delta)$ lattice code by $\mathcal{A}(n,\, k ,\, 2\delta)$. Assume we have a code of size $\mathcal{A}(n,\, k ,\, 2\delta)$ and construct the array $A$ from the codewords as in the proof of Theorem \ref{th:J1}. Consider the $(n-1,\, k-1 ,\, 2\delta)$ code obtained by deleting one of the elements of $\Pc$. It follows that the size of this code is $\leq \mathcal{A}(n-1,\, k-1 ,\, 2\delta)$ and we get one such code by the deletion of each of the $n$ elements of $\Pc$ in turn. Counting the number of non-zero entries of $A$ and comparing, we get the recursive relation: $k \, \mathcal{A}(n,\, k ,\, 2\delta) \leq n \, \mathcal{A}(n-1,\, k-1 ,\, 2\delta)$. Iterating the relation 
\begin{equation*}
\mathcal{A}(n,\, k ,\, 2\delta) \leq \biggl \lfloor (n/k)\,\mathcal{A}(n-1,\, k-1 ,\, 2\delta)\biggr\rfloor
\end{equation*}
 the bound follows.\\
 
\emph{Remark}: For a $ (7,\,5,\,2) $ constant weight lattice code, $ k^2 - kn + \delta n = 25 -35 +7 = -3 $, and so, applying the unrestricted Johnson bound we have $ N \leq 21 $. An optimal $ (7,\,5,\,2) $ constant weight code is the largest subgraph of the $ (7,\,5,\,1) $ Johnson graph, with adjacent vertices differing in exactly one element; which means it is actually the entire graph with $ 21 $ vertices. Hence, in this case, the unrestricted Johnson bound is achieved. For the $ (8,\,4,\,2) $ Johnson graph, there are $ 30 $ distinct maximal subgraphs, each with $ 14 $ vertices; this again achieves the unrestricted Johnson bound for an $ (8,\,4,\,4) $ constant weight code.
\subsection*{Optimal Constant Weight Codes}
As discussed before, a set of finite subsets of constituent elements of a suitable lattice is chosen as a constant weight code. Hence a constant weight code so constructed may be viewed as a subgraph of a generalized Johnson graph $J(n,\,k,\, i)$ and the code has minimum distance $ {\rm{d_{min}}} = 2i $; it follows that a \emph{maximal} subgraph in the generalized Johnson graph will be an optimal constant weight code. A \emph{clique} in a graph is a complete subgraph which is not contained in any other complete subgraph. Intuitively, an optimal code on a generalized Johnson graph should contain more members than a clique, as all the members of a clique are exactly at minimum distance from each other. However, it turns out that (Corollary \ref{cor:mclique}) for certain parameters an optimal code on the generalized Johnson graph can be realized as a maximal clique. Some optimal codes were obtained, for small parameters, by computer search implementing a variant of the well-known Bron-Kerbosch algorithm \cite{BK}. Additional codes were obtained by a two-step process: first, creating the generalized Johnson graph of all possible subsets and then searching for maximal cliques/ maximal subgraphs using the open software \textbf{Cliquer}. Some of the results are summarised in the given table.
\begin{table}
\begin{center}
\begin{tabular}{|l|l|l|l|}
\hline
$ (n, k, 2\delta) $ & $ J1 $ & $ J2 $ & $ N_{max} $ \\
\hline
$ (8,4,4) $ & -- & $ 14 $ &$ 30 $\\
$ (8,5,4)^{\ast} $ & $ 8/16 $ & --& $ 840 $\\
$ (9,4,4)^{\ast} $  & --& $ 18/25 $& $ 1708 $\\
$ (9,5,4) $ &$ 18 $ & -- & $ 1800 $\\
$ (9,7,4) $ &$ 4 $ & -- & $ 945 $\\
$ (9,6,6) $ &$ 3 $ & -- & $ 280 $\\
$ (10,3,4) $ & -- & $ 13 $ &$ 373680 $\\
$ (10,7,4)^{\ast} $ & -- & $ 8/22 $ &$ 3600 $\\
$ (10,6,6) $ &$ 5 $ & -- & $ 30240 $\\
$ (10,7,6) $ &$ 3 $ & -- & $ 2800 $\\
\hline
\end{tabular}
\end{center}
\caption{The parameters of some optimal codes obtained using Cliquer}
\end{table}
The columns $ J1 $ and $ J2 $ indicate whether the first (restricted) or the second (unrestricted) Johnson bound is applicable, respectively. A parameter set marked with the asterisk indicates a case where the Johnson bound is not achieved and then the values $ x/y$ in the columns $ J1 $ or $ J2 $ indicate the size of the maximal subgraphs obtained by Cliquer versus the computed bound. The last column $N_{max}$ gives the number of maximal cliques or subgraphs of the given size available. For instance, an optimal $(8,\,6,\,4)$ constant weight code is a maximal subgraph in the generalized Johnson graph $J(8,\,6,\, 2)$ comprising distinct $ 6 $-element subsets of a set of $ 8 $ elements such that every pair of distinct subsets differ in at least $ 2 $ places. There are $ 105 $ maximal subsets of size $ 4 $; choosing any of them will give an optimal $(8,\,6,\,4)$ constant weight code which achieves the first Johnson bound. For example, labelling the set of constituent elements as $ p_0,\, p_1, \cdots ,\, p_6, \, p_7 $, a constant weight code is given by $\lbrace \lbrace p_0,\, p_1,\, p_2, \, p_3, \, p_4,\, p_6 \rbrace , \, \lbrace p_0,\, p_1,\, p_2, \, p_5, \, p_6,\, p_7 \rbrace , \,\lbrace p_0,\, p_1,\, p_3, \, p_4, \, p_5,\, p_7 \rbrace ,$ $ \lbrace p_2,\, p_3,\, p_4, \, p_5, \, p_6,\, p_7 \rbrace \rbrace$.
\section{Lattice Codes for SAF Routing in Random Networks}\label{sec:saf}
In this section we give an application of constant weight codes over lattices with unique decomposition in error and erasure correction for SAF routing in random networks. We consider a single source communicating over a delay-free acyclic network to a single sink, i.e. a single unicast. There is no knowledge of the network topology assumed except for the following: at any intermediate node and at the sink, the number of incoming nodes is at most N.
\subsection{At the Source} 
As in the previous section, we begin with a source alphabet $\Ac$ which is a set of decomposable elements of a suitable lattice $\Lc$. Moreover, we choose each element of the source alphabet such that it admits a (primary or irreducible) decomposition into exactly $k$ constituent elements which do not divide one another. We denote the set of all such constituent elements by $\Pc$; let $\lvert \Pc \rvert = n$ and the source alphabet is uniquely identified with a subset of ${\Pc}_k$, the set of $ k $-subsets of $ \Pc $.\\
We choose a finite field $\bF_{q}$ with cardinality $q > n$ and define an injective map $f: \Pc \rightarrow \bF^{\times}_{q}, \, \, p \mapsto f(p)$, where $\bF_{q}^{\times} = \bF_{q}\backslash \{ 0 \}$. Next we define: 
\begin{equation*}
f^{k}:\Pc_k \rightarrow {(\bF^{\times}_{q})}^{k},\,\,(p_{1}, p_{2},\ldots,p_{k})\mapsto (f(p_{1}),f(p_{2}),\ldots,f(p_{k})).
\end{equation*}
At the source, the encoder performs this mapping and sends out the corresponding element of $(\bF^{\times}_{q})^{k} \subseteq \bF_{q}^{k}$, a row vector over $\bF_{q}$ of length $k$, as a single packet. 
\subsection{At the Intermediate Nodes} 
In one `generation', packets of $k \,\, q$-ary symbols arrive at each intermediate node from its incoming edges. Initially, all the incoming symbols are stored in a single array. If the number of incoming edges at each node is bounded above by N, the maximum length of the array required to store the incoming symbols during one transmission interval is N$k$. The intermediate node now extracts $k$ distinct symbols from the array, forms a packet and transmits it. The extraction of $k$ distinct symbols may be performed by constructing a binary search tree with repeated label rejection on the array of symbols: an algorithm of complexity $\mathcal{O} \left( k \ln k \right)$ \cite{Cor}. In the worst case scenario, the complexity will always be bounded by $\mathcal{O} \left( \text{N}k \right)$. If less than $k$ distinct symbols are extracted at any node during one generation, the node declares a failure and transmission is suppressed from that node into the network.
\subsection{At the Sink}
At the sink, the incoming symbols are again stored in an array and $k$ distinct elements are extracted in a similar manner as in the intermediate nodes. If less than $k$ distinct symbols are extracted at the sink, a $k$-symbol packet is created by adding the required number of zeros. Finally each symbol is operated on by the following map: \begin{equation*}
F^{-1}: \bF_q \rightarrow \Pc, \, x \mapsto 
\begin{cases}
f^{-1}(x)& \text{if}\,\, x \neq 0,\\
0& \text{otherwise}.
\end{cases}
\end{equation*}
From the resulting array, the non-zero elements are selected and the resulting set of constituent elements (of cardinality $\leq k$) is fed into the decoder which gives an estimate of the transmitted $k$-subset of constituent elements based on the symmetric distance metric as discussed in the previous section. Finally, the transmitted decomposable element of the source alphabet is retrieved by computing the meet of the estimated constituent elements.
\subsection{Error and Erasure Correction} To discuss the error and erasure correcting capability of the proposed code, we first define errors and erasures in the context of the proposed scheme. Recovering a complete set of $k$ distinct constituent elements at the sink which is different from the transmitted set constitutes an \emph{error} event, while recovering an incomplete set is an \emph{erasure} event. As discussed above, each codeword of a code is a set of distinct constituent elements of a suitable lattice. It follows that a constant weight code $\Cc$ based on constituent elements can correct upto $t$ errors in absence of erasures if and only if $\frac{1}{2} \,{\rm d_{min}}(\Cc)\, > 2t + 1$, i.e. any two distinct sets of constituent elements belonging to the code differ in at least $2t + 1$ places. However it can detect $t'$ errors if ${\rm d_{min}}(\Cc)\, > 2t' + 1$. But in the absence of errors, the code can correct $e$ erasures if and only if ${\rm d_{min}}(\Cc)\, > 2e + 1$. Hence it follows that the proposed coding scheme has asymmetric error and erasure correction capability - the ``cost'' associated with an error is twice that of an erasure.\\
\noindent For instance, suppose we are using the $(7,\,4,\,4)$ constant weight lattice code (described in Section \ref{sec:constwt}) for SAF routing. It follows that that the code can correct a single erasure in the absence of errors and can detect a single error. Let the set $\left \{ p_2,\, p_4,\, p_6,\, p_7 \right \}$ be the transmitted codeword. If the sink receives a set $\left \{ p_2,\, p_4,\, p_5,\, p_7 \right \}$ there are three contending codewords at the same distance from the received set and the decoder detects an error. If the sink receives an incomplete set with a single erasure, the decoding rule readily yields the transmitted codeword.\\
\emph{Remark}: As each codeword of an $(n,\,k,\,d)$ constant weight lattice code $\Cc$ is transmitted as a set of $k \,\, q$-ary symbols over the network, the normalized rate of the codeword may be defined as $R = \frac{{\rm log}_q(\lvert \Cc \rvert)}{k}$. Since $n$ can at most equal $q -1$, it follows that the code rate is in the interval $[0,\,1)$.
\section{Conclusion}
We have constructed constant weight codes on classes of lattices using the unique decomposition of lattice elements in terms of prime and irreducible elements. Primary element decomposition in Noether lattices is an abstraction of primary ideal decomposition in Noetherian commutative rings and the ideal lattice of a Noetherian ring is a modular Noether lattice. So we have achieved a generalization of the codes constructed in \cite{AG} on the ideals of Noetherian rings. It is seen that the existence of unique irreducible decomposition in lattices is not a special case of primary decomposition as in commutative rings but requires a different set of conditions. However, the unique prime decomposition of ideals in a Dedekind domain is an example of unique irreducible decomposition in lattices as well. Bounds constructed for constant weight codes in \cite{AG} have been restated in the setting of lattices, as the encoding exploits the properties of the Johnson graph on the finite subsets of a set of elements in both cases. We have also reformulated the application of codes based on ideals of Noetherian rings for store-and-forward (SAF) routing in random networks in terms of lattice-based codes.\\
As seen in Section \ref{sec:constwt}, optimal constant weight codes have been obtained as maximal subgraphs in the Johnson graph on the subsets of constituent elements of a suitable lattice. Moreover, when the restricted Johnson bound is achievable, an optimal code can be realized as a maximal clique of the generalized Johnson graph. For a number of small parameters, it has been found that the largest subgraph of the corresponding Johnson graph achieves the unrestricted bound as well. But it is seen that, in general, both the restricted and unrestricted Johnson bounds are loose. Hence a natural question is to relate the maximal clique/maximal subgraph with the sharper bounds. This might be achieved with an algebraic description for the optimal codes so obtained by search for maximal subgraphs on the Johnson graph on the subsets of constituent elements of lattices with unique decomposition. Finally, the construction of constant weight codes on multiplicative lattices presented in this paper have potential application in non-linear network coding for deterministic networks, where the networks can be modelled as suitable multiplicative lattices.
\section*{Acknowledgment}
The author gratefully acknowledges the kind assistance of Dilip P. Patil and Sumanta Mukherjee. Thanks also to Smarajit Das for his advice.

\medskip
\end{document}